\documentclass[onecolumn, a4size, 11pt]{IEEEtran}
\usepackage{amsmath}
\usepackage{amssymb}
\usepackage{amsfonts}
\usepackage{graphicx}
\usepackage{epsfig}
\usepackage{subfigure}
\usepackage{psfrag}
\usepackage{color}
\usepackage[noadjust]{cite}
\usepackage{multirow}
\usepackage{algorithm}
\usepackage{algpseudocode}
\usepackage{pifont}

\title{Distributed Energy Beamforming with \\ One-Bit Feedback }
\author{Seunghyun Lee and  Rui Zhang \\ ECE Department, 
National University of Singapore, Singapore \\ E-mail: s.lee@u.nus.edu, elezhang@nus.edu.sg}

\newtheorem{corollary}{\underline{Corollary}}[section]
\newtheorem{proposition}{\underline{Proposition}}[section]

\newtheorem{remark}{\underline{Remark}}[section]

\def\l{\left}
\def\r{\right}
\def\({\left(}
\def\){\right)}

\def\b0{{\mathbf{0}}}

\def\cA{\mathcal{A}}

\newcommand{\nn}{\nonumber}

\begin{document}
\maketitle \thispagestyle{empty}
\begin{abstract}
Energy beamforming (EB) is a key technique for achieving efficient radio-frequency (RF) transmission enabled wireless energy transfer (WET). By optimally designing the waveforms from multiple energy transmitters (ETs) over the wireless channels, they are constructively combined at the energy receiver (ER) to achieve an EB gain that scales with the number of ETs. However, the optimal design of transmit waveforms requires accurate channel state information (CSI) at the ETs, which is challenging to obtain in practical WET systems. In this paper, we propose a new channel training scheme to achieve optimal EB gain in a distributed WET system, where multiple separated ETs adjust their transmit phases to collaboratively send power to a single ER in an iterative manner, based on one-bit feedback from the ER per training interval which indicates the increase/decrease of the received power level from one particular ET over two preassigned transmit phases. The proposed EB algorithm can be efficiently implemented in practical WET systems even with a large number of distributed ETs, and is analytically shown to converge quickly to the optimal EB design as the number of feedback intervals per ET increases. Numerical results are provided to evaluate the performance of the proposed algorithm as compared to other distributed EB designs.  \end{abstract}


\setlength{\baselineskip}{1.3\baselineskip}

\section{Introduction}
Radio-frequency (RF) transmission enabled wireless energy transfer (WET) is a promising technique to achieve perpetually operating wireless devices by providing energy through the air. It has recently emerged as a new area of research in wireless communications, to charge communication nodes of typically low-power consumptions such as RF identification (RFID) devices and sensors distributed in a wide area (see, e.g., \cite{J_BHZ:2015, J_LWNKH:2015, A_BZZ} and the references therein). In particular, two appealing lines of research are  \emph{wireless powered communication network} (WPCN) \cite{J_JZ:2014} and \emph{simultaneous wireless information and power transfer} (SWIPT) \cite{J_ZH:2013}. In WPCN, wireless devices are powered by dedicated downlink WET for their uplink wireless information transmission (WIT); while in SWIPT, a dual use of RF signals is considered for simultaneous WET and WIT at the same time in the downlink. In both WPCN and SWIPT, efficient design of WET to compensate the significant power loss of RF signals over distance is essential.  

\begin{figure}
\centering
\subfigure[In-band WET]{
\centering
\includegraphics[width=7cm]{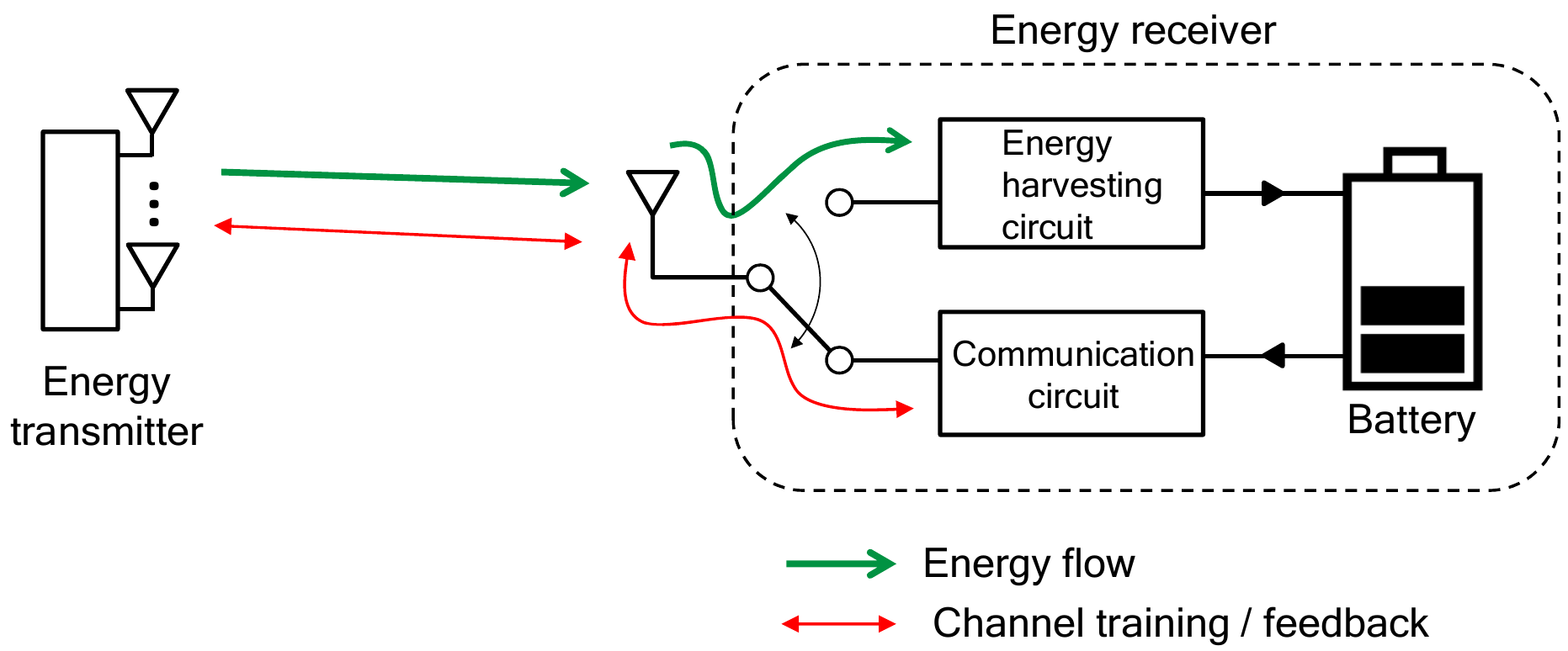}\label{Fig:InBand}} 
\subfigure[Out-band WET]{
\centering
\includegraphics[width=7cm]{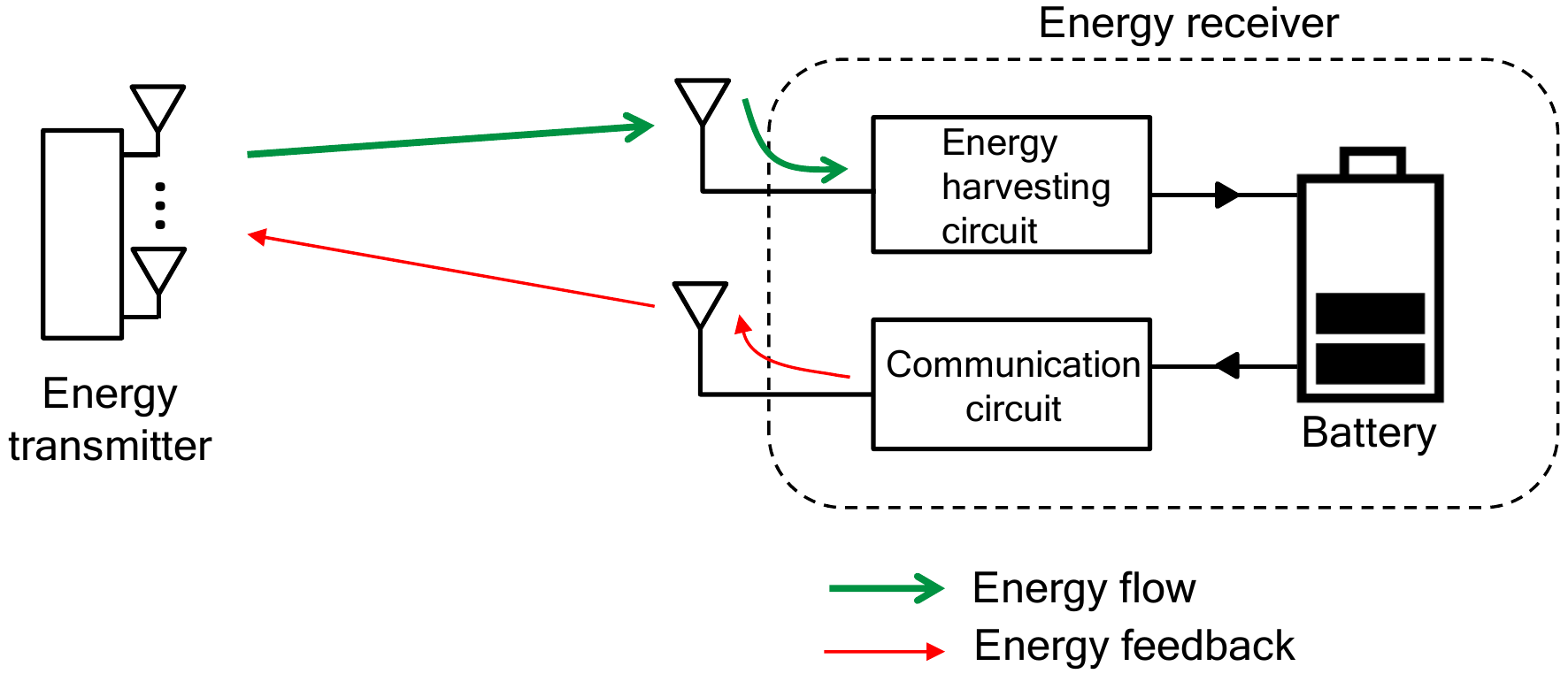}\label{Fig:OutBand}} 
\caption{Comparison of in-band and out-band WET.}
\label{Fig:ER}
\end{figure}

In practical WET systems, \emph{energy beamforming} (EB) is a key technique to significantly enhance the WET efficiency. With EB, the signal waveforms from multiple transmit antennas are optimally designed such that over different wireless channels they are constructively combined at a destined energy receiver (ER) to maximize the received signal amplitude or average power. However, in practice the energy transmitters (ETs)  need to acquire accurate knowledge of the channel state information (CSI) to maximize the EB gain. To this end, various CSI acquisition methods have been considered in the literature depending on the type of ER model used. In general, two practical ER models are considered for \emph{in-band} and \emph{out-band} WET and WIT, respectively, as shown in Fig.~\ref{Fig:ER}. For in-band WET, as shown in Fig.~\ref{Fig:InBand}, the WET and WIT are assumed to be implemented in the same frequency band, and  a single antenna is used at the ER for both energy harvesting and communication in a time-division-duplexing (TDD) manner. In this case, as in conventional wireless communication, the ET can send pilot signals to the ER that uses the communication circuit to estimate the channel and then send back the channel estimations to the ET for implementing EB \cite{J_YHG:2014}. However, this forward training approach incurs significant training and feedback overhead as the number of antennas at the ET increases. To overcome this difficulty, an alternative approach of reverse training is proposed in \cite{J_ZZ:2015, J_ZZ:2015_b}, where  training signals are sent by the ER to the ET to estimate the reverse-link channel that is assumed to be reciprocal of the forward-link channel over which the EB is implemented; as a result, the training overhead is independent of the number of antennas at the ET. On the other hand, for the case of out-band WET where the WET and WIT are implemented over different frequency bands, as shown in Fig.~\ref{Fig:OutBand}, two antennas are used at the ER for energy harvesting and communication over two orthogonal frequencies, respectively. Different from in-band WET, the out-band WET can be conducted with WIT at the same time, thus improving the efficiency. However, unlike the in-band WET case, the ET cannot obtain the channel to the energy harvesting antenna at the ER by conventional forward/reverse training methods as the communication antenna at the ER operates at a different frequency from that for the energy harvesting antenna in the out-band WET case. Therefore, a new channel estimation method based on the feedback of the measured power level at the energy harvesting antenna of the ET is proposed in \cite{J_XZ:2014, A_XZ:2015}. In this method, the energy feedback from the ER is used for iteratively localizing the target multiple-input multiple-output (MIMO) channel by applying the cutting-plane method in convex optimization. 

\begin{figure} 
\centering
\includegraphics[width=7cm]{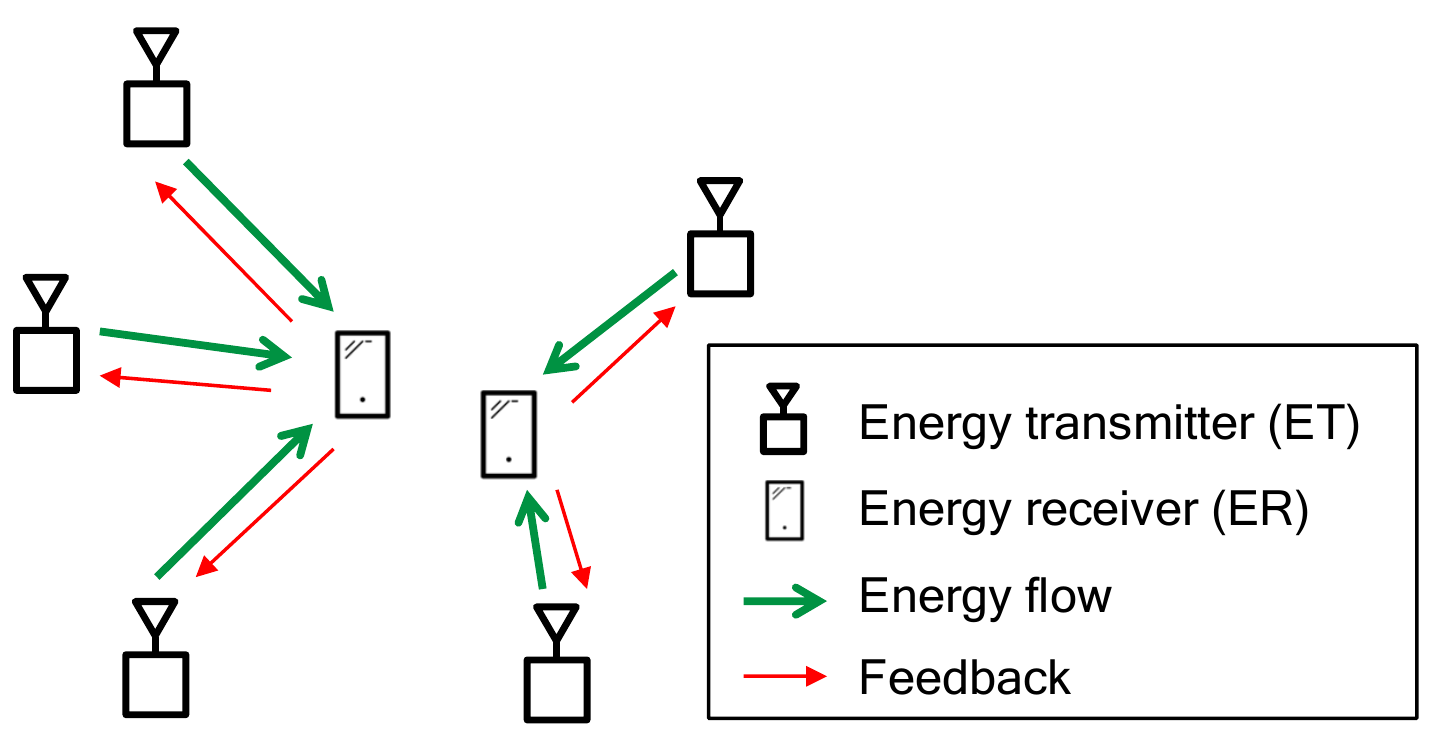}
\caption{A distributed WET system with ER feedback for channel training.} \label{Fig:SystemModel}
\end{figure}

In this paper, we consider the optimal EB design in a distributed WET system with multiple separated single-antenna ETs that cooperatively send power to one or more single-antenna ERs under the out-band WET model, as illustrated in Fig.~\ref{Fig:SystemModel}. Note that different from \cite{J_XZ:2014, A_XZ:2015}, where all antennas are equipped at one single ET and thus their channels to the ER can be jointly estimated, the distributed EB considered in this paper needs to be implemented over different ETs without the need of their centralized processing. In practice, EB by distributed ETs has the advantage of avoiding high power intensity from any single ET to the ER as compared to the conventional EB by a single multi-antenna ET \cite{J_ZH:2013, J_YHG:2014, J_ZZ:2015, J_ZZ:2015_b, J_XZ:2014, A_XZ:2015}, thus significantly improving its safety in operation. Motivated by this, in this paper we propose a new training scheme for distributed EB. First, we propose a new phase adaptation algorithm for a single ET to tune the phase of its transmit waveform to optimally align to the phase of the sum-signal from all other ETs (each assumed to transmit with a fixed phase) at the ER. Specifically, this ET transmits with a pair of two alternating phases at each feedback interval to iteratively find the optimal phase based on one-bit feedback from the ER, indicating which one of the two phases used at each interval results in larger harvested power than the other. It is shown that, with each feedback bit, the algorithm can reduce the size of the target set containing the optimal phase by half, and thus can converge to the optimal phase exponentially fast with the increasing number of feedback intervals. Based on this new phase adaptation algorithm, a distributed EB protocol is then proposed where the ETs in the system take turns to sequentially adapt their transmit phases with those of all other active ETs being fixed, which is shown to efficiently converge to the optimal EB solution as the number of per-ET feedback intervals increases. 

It is worth noting an existing one-bit feedback scheme for distributed beamforming proposed in \cite{J_MHMB:2010}, referred to as \emph{random phase perturbation} scheme in this paper. In this scheme, transmitters independently adjust their transmit phases via random perturbation based on one-bit feedback from the receiver, which indicates the increase or decrease of the current signal-to-noise-ratio (SNR) as compared to its recorded highest SNR. This scheme is also applicable to distributed EB for WET. By simulation, it is shown that our proposed distributed EB scheme outperforms that in \cite{J_MHMB:2010} in terms of convergence speed to the optimal EB solution.

\section{System Model}
As shown in Fig.~\ref{Fig:SystemModel}, we consider a distributed WET system where $M>1$ single-antenna ETs collaboratively send wireless energy to  $K\geq 1$ single-antenna ERs. In this paper, we focus on the special case of a single ER, i.e., $K=1$, while the general case of $K>1$ ERs will be considered in the journal version of this paper. For the single ER, we adopt the out-band WET model in Fig.~\ref{Fig:OutBand}, and assume that it can send energy feedback to all ETs at a given frequency different from that for the WET.

Since energy signals carry no information, for simplicity we assume the transmit signal of ET$_m$, $m=1,...,M$, to be an unmodulated carrier signal with phase offset $\phi_m\in [-\pi, \pi)$, which is expressed as
 \begin{equation} \label{Eq:TxSignal1}
s_m(t) =  \sqrt{2P}\cos(2\pi f_c t + \phi_m), 
\end{equation}
where $P$ denotes the transmit power of each ET, and $f_c$ is the carrier frequency. Each transmitted signal propagates through a  multi-path wireless channel in general. Thus, the received signal $r(t)$ at  ER is expressed as 
\begin{align} \label{Eq:RxSignal1}
r(t) = \sum_{m=1}^M \sum_{l=1}^{L_m} a_{m,l} \sqrt{2P}  \cos(2\pi f_c (t - \tau_{m,l}) + \phi_m),
\end{align}
where $L_m$ is the number of signal paths from ET$_m$ to ER, and $a_{m,l}$, $\tau_{m,l}$ are the signal attenuation and delay of the $l$th path, respectively, with $l=1,...,L_m$. The received signal given in \eqref{Eq:RxSignal1} can be simplified as
\begin{equation} \label{Eq:RxSignal2}
r(t) = \sqrt{2P}\sum_{m=1}^M \sqrt{\beta_m} \cos(2\pi f_c t + \phi_m - \theta_m),
\end{equation}
where $\beta_m$ and $\theta_m$ are the overall power gain and phase shift of the multi-path channel from ET$_m$ to ER, respectively, given by
\begin{equation*}
\beta_m = \l(\sum_{l=1}^{L_m} a_{m,l}\cos(2\pi f_c \tau_{m,l})\r)^2 + \l(\sum_{l=1}^{L_m} a_{m,l}\sin(2\pi f_c \tau_{m,l})\r)^2, 
\end{equation*}
and 
\begin{equation*}
\theta_m = \arctan\l(\frac{\sum_{l=1}^{L_m} a_{m,l}\sin(2\pi f_c \tau_{m,l})}{\sum_{l=1}^{L_m} a_{m,l}\cos(2\pi f_c \tau_{m,l})}\r).
\end{equation*}
The average harvested power at ER, denoted by $Q$, is then given by
\begin{align}
Q &=  \frac{\rho}{T}\int_{0}^{T} |r(t)|^2 dt \nn\\
& = \rho P \bigg(\sum_{m=1}^M \beta_m  
 + \sum_{i,j=1, i\neq j}^M  \sqrt{\beta_i \beta_j} \cos((\phi_i - \theta_i) - (\phi_j - \theta_j))     \bigg), \label{Eq:P_r}
\end{align}
where $0<\rho \leq 1$ is the energy conversion efficiency at ER, and we assume $\rho=1$ in the sequel for notational convenience; and $T = \frac{1}{f_c}$ is the period of the carrier signal.

If each ET$_m$ perfectly knows the phase shift of its channel $\theta_m$, the optimal transmit phase that maximizes the harvested power $Q$ in \eqref{Eq:P_r} is given by $\phi_m^\star = \theta_m$, $m=1,...,M$. We refer to this case as the \emph{optimal EB} in the sequel. With the optimal EB, the maximum harvested power at ER, denoted by $Q^\star$, is thus given by
\begin{equation} \label{Eq:P_r max}
Q^\star = P\l(\sum_{m=1}^M \beta_m + \sum_{i,j=1, i\neq j}^M\sqrt{\beta_i \beta_j}     \r).
\end{equation}
In practice, only imperfect CSI is available at each ET, and thus the maximum harvested power in \eqref{Eq:P_r max} with the optimal EB only provides a performance upper bound for practical distributed WET systems. In this paper, we propose a new training scheme for practical distributed WET systems to efficiently implement the optimal EB with low complexity and feedback  overhead. First, in Section~\ref{Section:PhaseAdaptation}, we present a new phase adaptation algorithm for a single ET to iteratively adapt its transmit phase to optimally align to the phase of the sum-signal received at ER from other active ETs, based on one-bit energy feedback from ER. Next, in Section~\ref{Section:Protocol}, we present a distributed EB protocol where the ETs sequentially adapt their transmit phases using the phase adaptation algorithm proposed in Section~\ref{Section:PhaseAdaptation}. 

\section{Phase Adaptation Algorithm for a Single ET} \label{Section:PhaseAdaptation}

\begin{figure} 
\centering
\includegraphics[width=6cm]{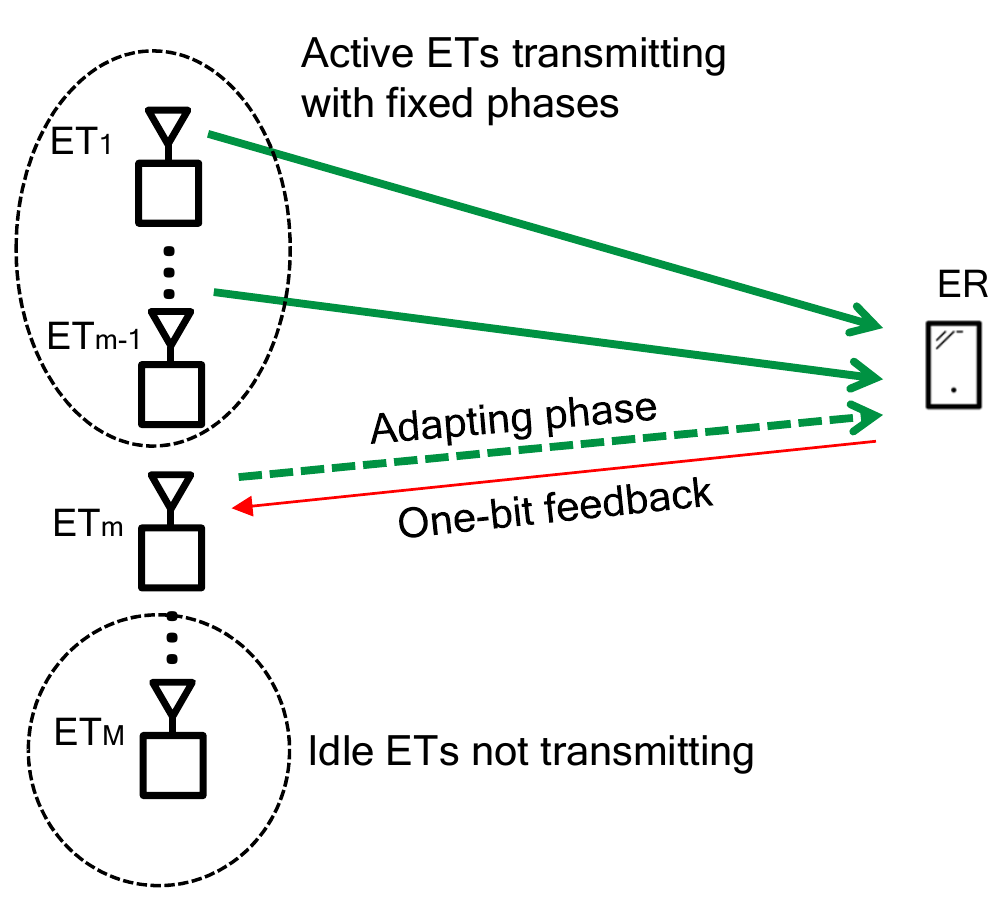}
\caption{Illustration of the transmit phase adaptation by ET$_m$ based on one-bit feedback from ER.} \label{Fig:Adapt}
\end{figure}

Suppose that ET$_1$, ET$_2$,$...$, and ET$_{m-1}$, are transmitting with fixed phases $\phi_1 = \bar{\phi}_1$, $\phi_2 = \bar{\phi}_2,...$, $\phi_{m-1} = \bar{\phi}_{m-1}$, respectively, whereas ET$_{m+1},...,$ ET$_M$, are idle (not transmitting), as shown in Fig.~\ref{Fig:Adapt}. In this section, we present an algorithm for  ET$_m$ to iteratively adjust its transmit phase $\phi_m$ so that the received signal can be coherently added to the sum-signal from all the other $m-1$ active ETs, based on one-bit energy feedback from ER. 

For the purpose of exposition, we express the received signal at ER from all active ETs (including ET$_m$) as
\begin{align} \label{Eq:RxSignal3}
r_m(t) = \sqrt{2P}\bigg(\sqrt{\beta_m} \cos(2\pi f_c t + \phi_m - \theta_m)  + \sum_{i=1}^{m-1} \sqrt{\beta_i} \cos(2\pi f_c t + \bar{\phi}_i - \theta_i) \bigg).
\end{align}
Let $Q_m(\phi_m)$ denote the average harvested power at ER with transmit phase $\phi_m$ of ET$_m$. We thus have
\begin{align}
Q_m(\phi_m) &=  \frac{1}{T}\int_{0}^{T} |r_m(t)|^2 dt \nn\\
& = P\l(\beta_m + \alpha + 2\sqrt{\beta_m \alpha} \cos(\phi_m - ( \theta_m - \varphi))\r), \label{Eq:P_r i} 
\end{align}
where 
\begin{equation*}
\alpha \triangleq \l(\sum_{i=1}^{m-1}\sqrt{\beta_i}\cos(\bar{\phi}_i - \theta_i) \r)^2 + \l(\sum_{i=1}^{m-1}\sqrt{\beta_i}\sin(\bar{\phi}_i - \theta_i) \r)^2,
\end{equation*}
and 
\begin{equation*}
\varphi \triangleq \arctan\l(\frac{\sum_{i=1}^{m-1}\sqrt{\beta_i}\sin(\bar{\phi}_i - \theta_i)}{\sum_{i=1}^{m-1}\sqrt{\beta_i}\cos(\bar{\phi}_i - \theta_i)} \r).
\end{equation*}

Let $\phi_m^*$ denote the optimal transmit phase of ET$_m$ to maximize $Q_m(\phi_m)$ given in \eqref{Eq:P_r i}, which is given by $\phi_m^* = \theta_m - \varphi$. The proposed algorithm for ET$_m$'s phase adaptation to converge to $\phi_m^*$ is explained as follows.  As illustrated in Fig.~\ref{Fig:Algorithm}, in each training interval, ET$_m$ sequentially transmits to ER with a pair of two preassigned phases $\psi$ and $\psi'$, respectively, each with duration $T_s$,\footnote{In practice, $T_s$ needs to be sufficiently large so that ER can accurately measure the harvested power with each transmit phase. For example, $T_s$ can be set as a certain number of periods of the carrier signal, i.e., $T_s = k T$ where $k$ is a sufficiently large integer.} i.e., each interval is of total duration $2T_s$. The ER measures the average power received corresponding to each of the two phases in each interval and at the end of the interval it sends back to ET$_m$ one bit  indicating whether $Q_m(\psi)> Q_m(\psi')$ or $Q_m(\psi)< Q_m(\psi')$, or equivalently $\cos(\psi - \phi_m^*) > \cos(\psi' - \phi_m^*)$ or $\cos(\psi - \phi_m^*) < \cos(\psi' - \phi_m^*)$, according to \eqref{Eq:P_r i}. Notice that for simplicity in this paper we assume that the power measurement at ER is perfect at each training/feedback interval.  Let $\cA^{(n)}\subseteq [-\pi,\pi)$ denote the working set of ET$_m$ at the $n$th feedback interval, where $\phi_m^*\in \cA^{(n)}$. If  $\cos(\psi - \phi_m^*) > \cos(\psi' - \phi_m^*)$ at the $n$th feedback interval, ET$_m$ can infer that the target phase $\phi_m^*$ is located in $ \cA^{(n)}\backslash \{\theta\in [-\pi,\pi): \cos(\theta - \psi) < \cos(\theta - \psi')\}$, and vice versa; as a result, ET$_m$ can reduce its working set in half (see Fig.~\ref{Fig:Algorithm}). At the next feedback interval, ET$_m$ updates $\psi$ and $\psi'$ to be $\psi = \max_{\theta\in \cA^{(n)}} \theta$ and $\psi' = \min_{\theta\in \cA^{(n+1)}} \theta$, i.e., the two boundary phase values of the current working set. The above procedure is repeated for $N$ feedback intervals, and finally ET$_m$ sets $\bar{\phi}_m = \frac{\psi + \psi'}{2}$. This phase adaptation algorithm, referred to as (A1), is summarized in \textbf{Algorithm 1}.

\begin{figure} 
\centering
\includegraphics[width=11cm]{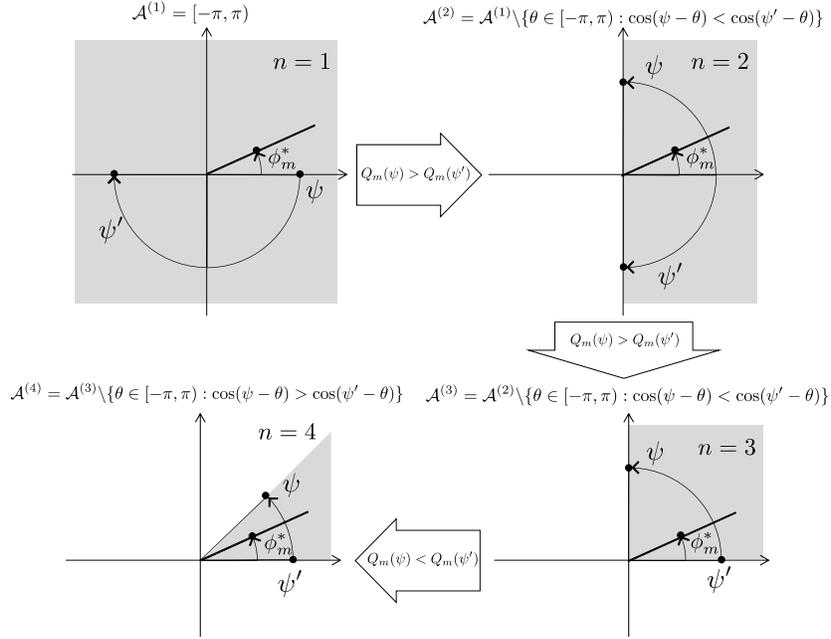}
\caption{Illustration of the proposed phase adaptation algorithm for ET$_m$ to find the target phase $\phi_m^*$, where the shaded regions indicate the working set $\cA^{(n)}$ at the beginning of the $n$th feedback interval. } \label{Fig:Algorithm}
\end{figure}

\begin{algorithm}
\caption{(A1): Phase Adaptation Algorithm for ET$_m$}
\begin{algorithmic}[1]
\State \textbf{Initialize:} ET$_m$ sets $\cA^{(1)} = [-\pi, \pi)$, $\psi = 0$, and $\psi'  = -\pi$.
\For{ $n=1 : N$}
\State ET$_m$ transmits with $\phi_m = \psi$ and then $\phi_m = \psi'$.
\State ER compares $Q_m(\psi)$ and $Q_m(\psi')$ and feeds back the corresponding one-bit to ET$_m$.
\If{ $Q_m(\psi) > Q_m(\psi')$}
\State ET$_m$ sets $\cA^{(n+1)} = \cA^{(n)}\backslash \{(\theta\in [-\pi,\pi): \cos(\theta - \psi) < \cos(\theta - \psi')\}$.
\Else
\State ET$_m$ sets $\cA^{(n+1)} = \cA^{(n)}\backslash \{(\theta\in [-\pi,\pi): \cos(\theta - \psi) > \cos(\theta - \psi')\}$.
\EndIf
\State ET$_m$ sets $\psi = \max_{\theta\in \cA^{(n+1)}} \theta$ and $\psi' = \min_{\theta\in \cA^{(n+1)}} \theta$. 
\EndFor
\State ET$_m$ sets $\bar{\phi}_m = \frac{\psi + \psi'}{2}$.
\end{algorithmic}
\end{algorithm}

\begin{remark} \label{Remark:WorstCase}
We define the phase-error between the phase determined by (A1) at ET$_m$, i.e., $\bar{\phi}_m$, and the optimal phase $\phi_m^*$ as
\begin{equation} \label{Eq:Error}
e_m = \bar{\phi}_m - \phi_m^* = \bar{\phi}_m - (\theta_m - \varphi).
\end{equation}
Since the working set is divided into two (i.e., bisected) after each feedback interval and we choose $\bar{\phi}_m = \frac{\psi + \psi'}{2}$ at the end of the algorithm, it can be shown that the absolute value of the phase-error at ET$_m$ is upper-bounded by
\begin{equation} \label{Eq:ErrorBound}
|e_m| \leq \frac{\pi}{2^N}.
\end{equation}
In other words, after $N$ feedback intervals, the worst-case error between the estimated phase and the optimal phase is no larger than $2^{-N}\pi$, which exponentially decreases to zero with increasing $N$.
\end{remark}

\section{Distributed EB} \label{Section:Protocol}
In this section, we present a distributed EB scheme based on (A1) proposed in Section~\ref{Section:PhaseAdaptation}, for all $M$ ETs to collaboratively send power to ER. 

\subsection{Distributed Protocol} \label{Subsection:Protocol}
The proposed distributed EB protocol is shown in Fig.~\ref{Fig:Protocol}, and explained as follows. 
\begin{itemize}
\item To initiate distributed EB, ER sends a ``start" signal to all $M$ ETs.\footnote{In practice, ER may send the ``start" signal when the harvested power becomes below some predefined threshold due to e.g., channel variations. }

\item Once the ETs receive the ``start" signal, they stop transmitting, except ET$_1$ which transmits with an arbitrary fixed phase $\bar{\phi}_1$. Without loss of generality, we assume $\bar{\phi}_1 = 0$. 

\item ET$_2$ adapts its phase via (A1) for $N$ feedback intervals to match  ET$_1$'s signal at ER. After the adaptation, it keeps transmitting with the determined phase $\bar{\phi}_2$.

\item ET$_3$ adapts its phase via (A1) for $N$ feedback intervals to match ET$_1$'s and ET$_2$'s signals at ER. After the adaptation, it keeps transmitting with the determined phase $\bar{\phi}_3$.

\item The above procedure is repeated sequentially, until ET$_M$.
\end{itemize}

After all $M$ ETs set their phases to be $\phi_m = \bar{\phi}_m$, $m=1,...,M$, via the above distributed EB protocol, the harvested power at ER, denoted by $Q_\text{d}$, is given by (via substituting $\phi_m=\bar{\phi}_m$, $m=1,...,M$, into \eqref{Eq:P_r})
\begin{align} \label{Eq:P_r Protocol}
Q_\text{d} = P\bigg(\sum_{m=1}^M \beta_m   +  \sum_{i,j=1, i\neq j}^M\sqrt{\beta_i \beta_j}\cos((\bar{\phi}_i - \theta_i) - (\bar{\phi}_j - \theta_j))\bigg). 
\end{align}
Note that $Q_\text{d}\leq Q^\star$ in general according to \eqref{Eq:P_r max}. In the following, we analyze the performance of the proposed distributed EB scheme.

\begin{figure} 
\centering
\includegraphics[width=10cm]{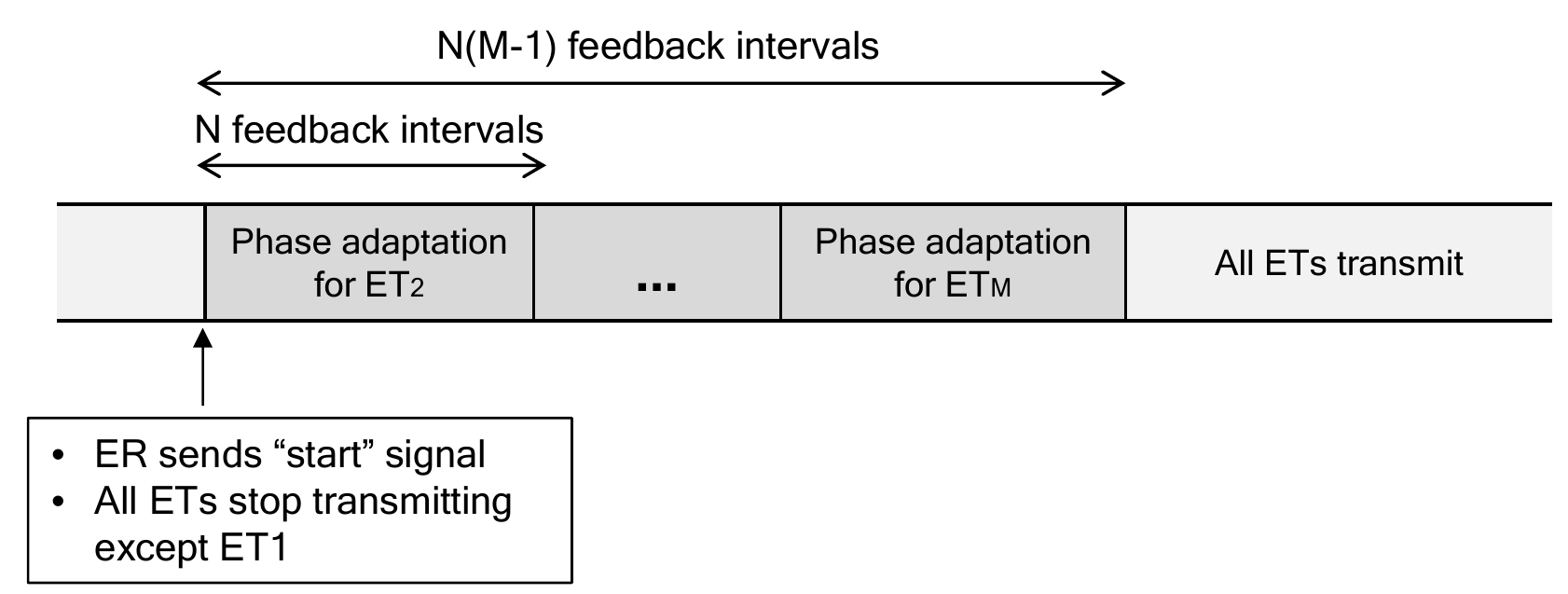}
\caption{The proposed distributed EB protocol. }\label{Fig:Protocol}
\end{figure}

\subsection{Performance Analysis}
To analyze the performance of the proposed scheme, we define the resulting efficiency of the distributed WET system, denoted by $\eta$, as the ratio between $Q_\text{d}$ and $Q^\star$ given in \eqref{Eq:P_r Protocol} and \eqref{Eq:P_r max}, respectively, i.e.,
\begin{equation} \label{Eq:Efficiency}
\eta = \frac{Q_\text{d}}{Q^\star},
\end{equation}
where $0<\eta\leq 1$ since $Q_\text{d}\leq Q^\star$ in general. We first present a lower bound on $\eta$ in the following proposition.

\begin{proposition} \label{Proposition:LowerBound}
The efficiency of the proposed distributed EB scheme defined in \eqref{Eq:Efficiency} is lower-bounded by
\begin{equation} \label{Eq:EfficiencyLower}
\eta \geq \frac{1}{Q^\star}\l(\sum_{m=1}^M \beta_m + \sum_{i,j=1, i\neq j}^M \sqrt{\beta_i \beta_j}\cos^2\l(\frac{\pi}{2^N}\r) \r)
\end{equation}
\end{proposition}
\begin{proof}
See the Appendix.
\end{proof}

It can be observed from \eqref{Eq:EfficiencyLower} that as the number of per-ET feedback intervals in (A1), $N$, becomes large, the lower bound approaches one. In other words, for the ideal case of $e_m = 0$, $m=1,...,M$, which is obtained when $N\rightarrow \infty$, the proposed protocol achieves the maximum harvested power by the optimal EB given in \eqref{Eq:P_r max}. 

Next, we analyze the required number of per-ET feedback intervals for (A1) to achieve a given target efficiency, denoted by $0<\hat{\eta}\leq 1$. By re-arranging the terms in the inequality $\eta\geq \hat{\eta}$ and using \eqref{Eq:EfficiencyLower}, we obtain the following corollary. 

\begin{corollary}
If the number of per-ET feedback intervals in (A1) satisfies 
\begin{equation} \label{Eq:WorstCaseBound}
N \geq  \log_2\l(\frac{\pi}{\arccos\l(\sqrt{\hat{\eta} - \frac{(1-\hat{\eta})\sum_{m=1}^M \beta_m}{\sum_{i,j=1, i\neq j}^M \sqrt{\beta_i \beta_j}} }\r)}\r),
\end{equation}
then it holds that $\eta\geq\hat{\eta}$.
\end{corollary}

It is worth noting from \eqref{Eq:WorstCaseBound} that for the special case of $\beta_1 = \beta_2 = ... = \beta_K = \beta$ (i.e., all ETs have identical channel gains to ER), we have $\sum_{m=1}^M \beta_m = M\beta$ and $\sum_{i,j=1, i\neq j}^M \sqrt{\beta_i \beta_j} = M(M-1)\beta$. As a result, \eqref{Eq:WorstCaseBound} becomes
\begin{equation} \label{Eq:WorstCaseBoundBeta}
N \geq  \log_2\l(\frac{\pi}{\arccos\l(\sqrt{ \frac{M\hat{\eta} - 1}{M-1}}\r)}\r) .
\end{equation}
For instance, if $M=5$, i.e., if there are five ETs, \eqref{Eq:WorstCaseBoundBeta} yields $N\geq4.8094$ with $\hat{\eta} = 0.99$ and $N\geq 6.4731$ with $\hat{\eta} = 0.999$, respectively. Thus, each ET needs $N=5$ and $N=7$ feedback intervals with (A1) to ensure 99\% and 99.9\% of the optimal EB gain achieved by the proposed distributed EB scheme, respectively.

\section{Numerical Results}
In this section, we evaluate the performance of the proposed distributed EB scheme by simulation. For the simulation, we set the transmit power of each ET as $P=1$ Watt (W), the number of signal paths between each ET$_m$ and ER as $L_m = 1$ (which corresponds to the line of sight (LoS) environment). Moreover, the channel power gain $\beta_m$ is modeled by  path-loss only, given by $\beta_m = c_0 (r_m/r_0)^{-\delta}$, where $c_0 = -20$ dB is a constant attenuation due to the path-loss at a reference distance $r_0 = 1$ meter (m), $\delta = 3$ is the path-loss exponent, and $r_m$ is the distance between ET$_m$ and the ER. We choose the distribution of the distance $r_m$ and the random phase shift $\theta_m$ of each ET$_m$ to be  $r_m\sim \text{Uniform}(5,15)$ (in meters) and $\theta_m\sim\text{Uniform}(-\pi,\pi)$, respectively. 

\begin{figure} 
\centering
\includegraphics[width=8cm]{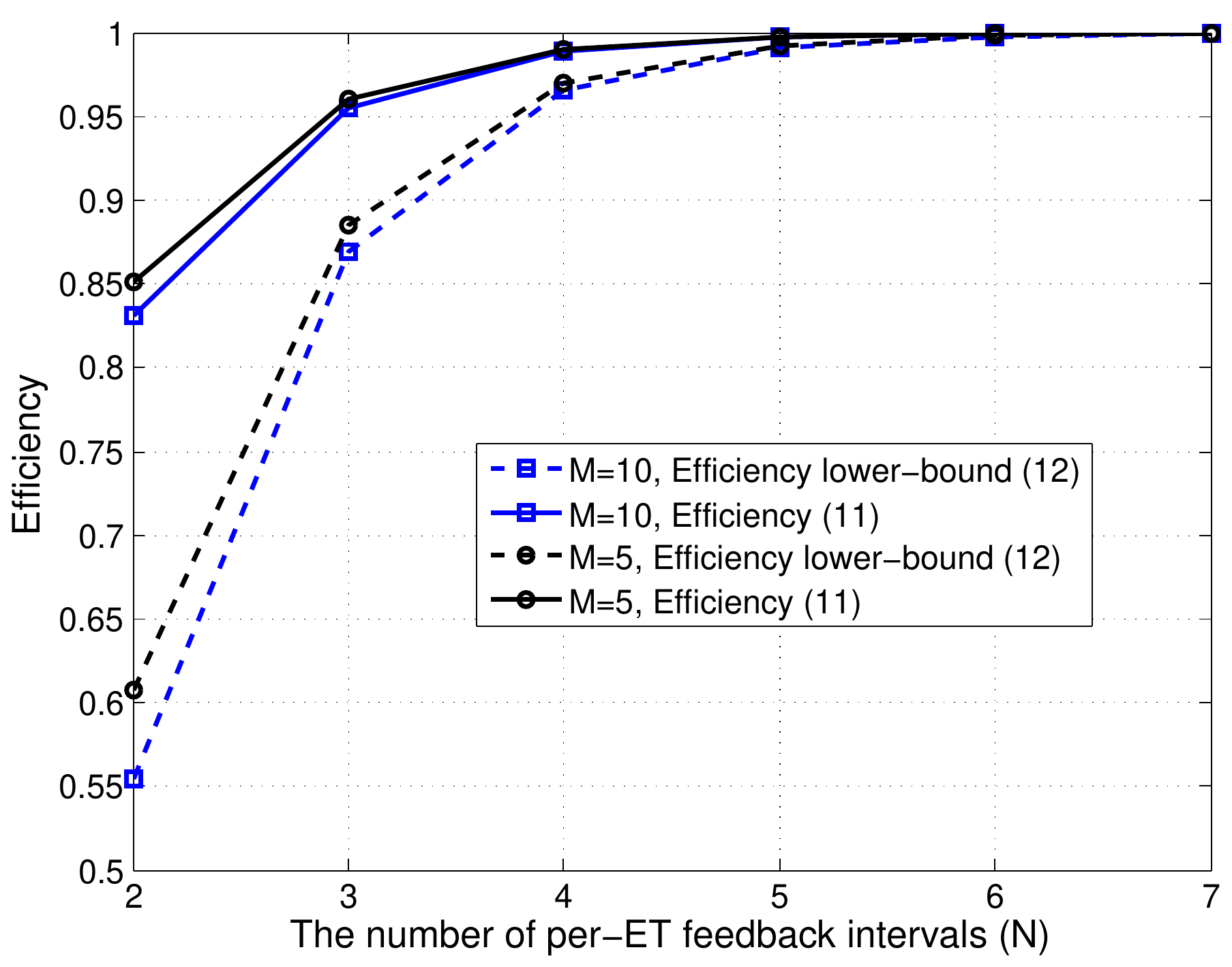}
\caption{The efficiency of the proposed distributed EB scheme, $\eta$, versus the number of per-ET feedback intervals, $N$. } \label{Fig:Efficiency}
\end{figure}

Under the above setup, Fig.~\ref{Fig:Efficiency} shows the efficiency of the proposed distributed EB scheme given in \eqref{Eq:Efficiency} and its lower bound in \eqref{Eq:EfficiencyLower} for the cases of $M=5$ and $M=10$, by averaging over 5000 randomly generated $r_m$ and $\theta_m$, $m=1,...,M$. First, it is observed that for both cases of $M=5$ and $M=10$,  the efficiency lower-bound given in \eqref{Eq:EfficiencyLower} becomes tighter as the number of per-ET feedback intervals $N$ increases, and eventually converges to 1, as compared to the exact efficiency in \eqref{Eq:Efficiency}. Furthermore, it can be seen from Fig.~\ref{Fig:Efficiency} that $N=4$ already results in the efficiency higher than 0.95 for both cases of $M=5$ and $M=10$.

\begin{figure} 
\centering
\includegraphics[width=8cm]{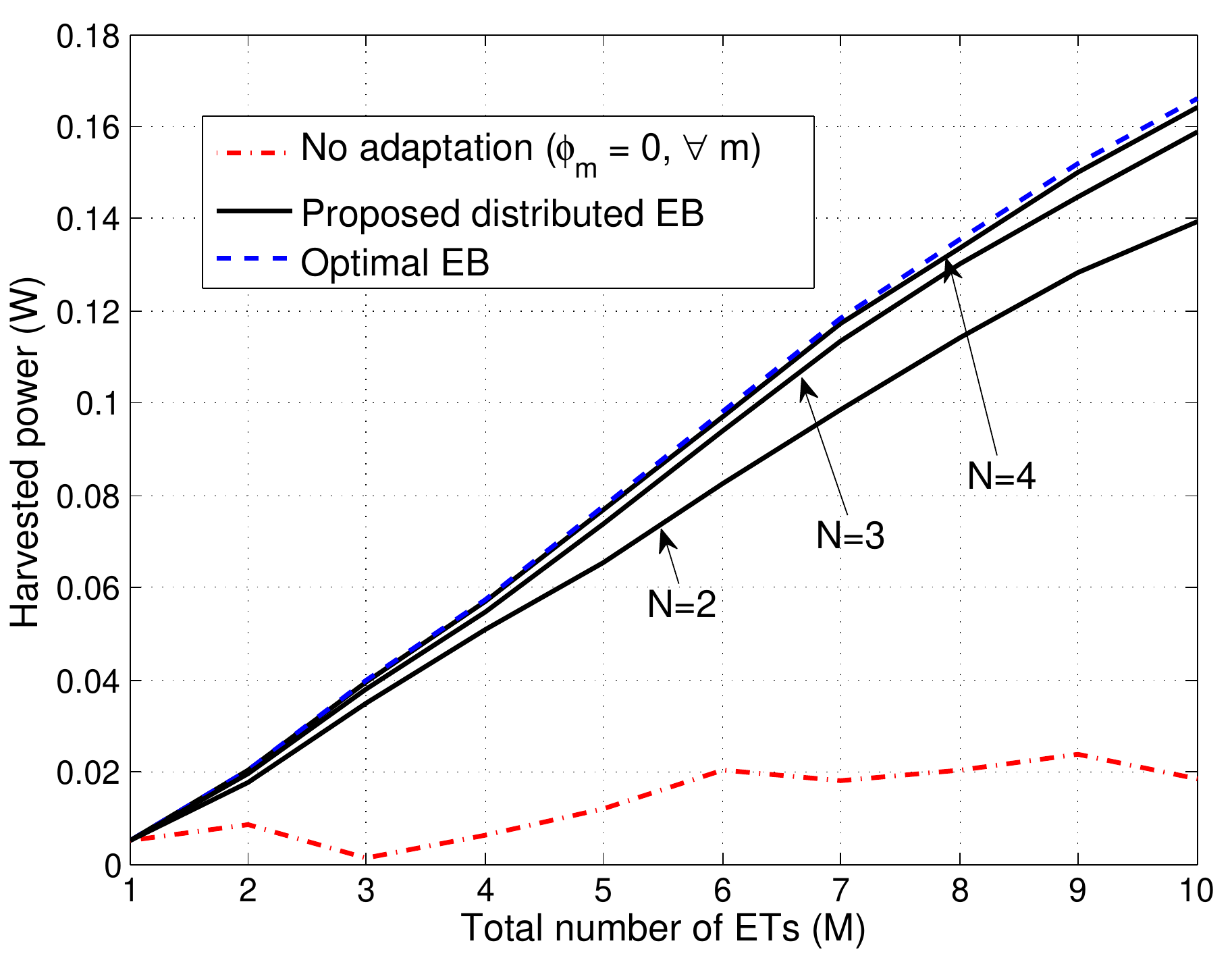}
\caption{The harvested power versus the total number of ETs, $M$. } \label{Fig:EBgain}
\end{figure}

In Fig.~\ref{Fig:EBgain}, we compare the harvested power by ER versus the total number of ETs $M$ with different numbers of per-ET feedback intervals $N$, for one set of random realizations of $r_m$ and $\theta_m$, $m=1,...,M$. First, it is observed that the harvested power of the proposed distributed EB scheme keeps increasing with $M$, thanks to the phase adaptation of each ET via (A1); whereas that of no adaptation, in which each ET$_m$ fixes its phase to be $\bar{\phi}_m = 0$ at all time, fluctuates over $M$ in general, which is due to the fact that the channel phases from ETs to ER are different and as a result, their received signals may add constructively or destructively at ER. Second, it can be seen from Fig,~\ref{Fig:EBgain} that, the larger the number of per-EB feedback intervals $N$, the higher is the harvested power achieved by the proposed distributed EB scheme, which is expected since larger $N$ yields more accurate estimated phase $\phi_m^*$ for each ET$_m$ via algorithm (A1). 

\begin{figure} 
\centering
\includegraphics[width=7.5cm]{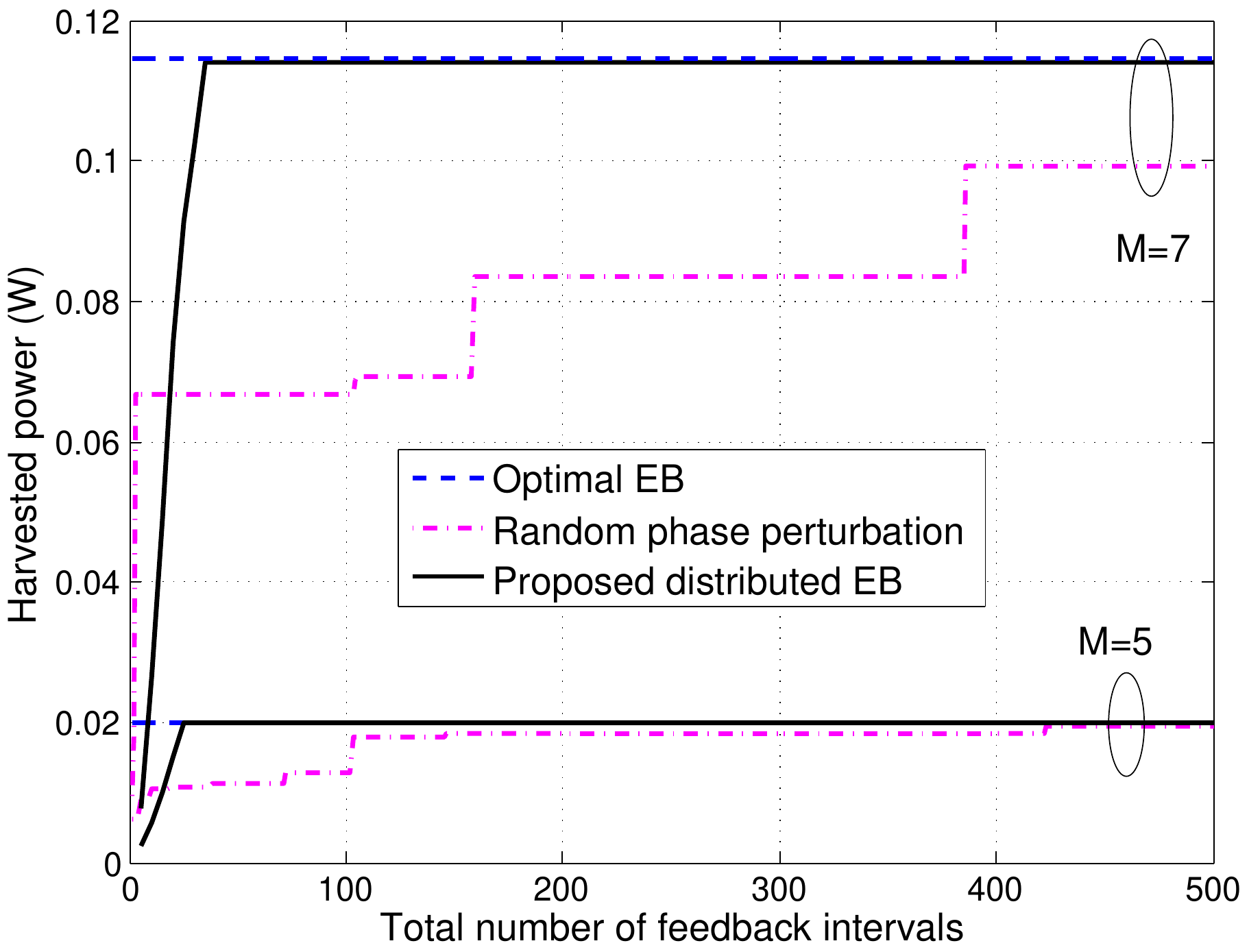}
\caption{The harvested power versus the total number of feedback intervals. } \label{Fig:Comparison}
\end{figure}

Next, in Fig.~\ref{Fig:Comparison}, we compare the convergence performance of our proposed distributed EB scheme with the random phase perturbation scheme proposed in \cite{J_MHMB:2010}, for the cases of $M=5$ and $M=7$, respectively, for one set of random realizations of $r_m$ and $\theta_m$, $m=1,...,M$. We set the number of per-ET feedback intervals to be $N=5$ for our proposed scheme. As a result, our proposed scheme requires $N(M-1)= 20$ feedback intervals in total for the case of $M=5$ and $30$ feedback intervals in total for the case of $M=7$, in order for all ETs to set their phases $\bar{\phi}_m$, $m=1,...,M$  (see Fig.~\ref{Fig:Protocol}). However, the random phase perturbation scheme needs more than 100 feedback intervals for the harvested power to converge, as observed from Fig.~\ref{Fig:Comparison}. This is because,  in our proposed scheme, the phase adaptation of each ET requires only a fixed number of feedback intervals $N$, whereas in the random phase perturbation scheme, all ETs randomly perturb their phases at the same time in each interval, and as a result it takes longer time for the transmit phases of all ETs to converge in general. 

\begin{figure} 
\centering
\includegraphics[width=8cm]{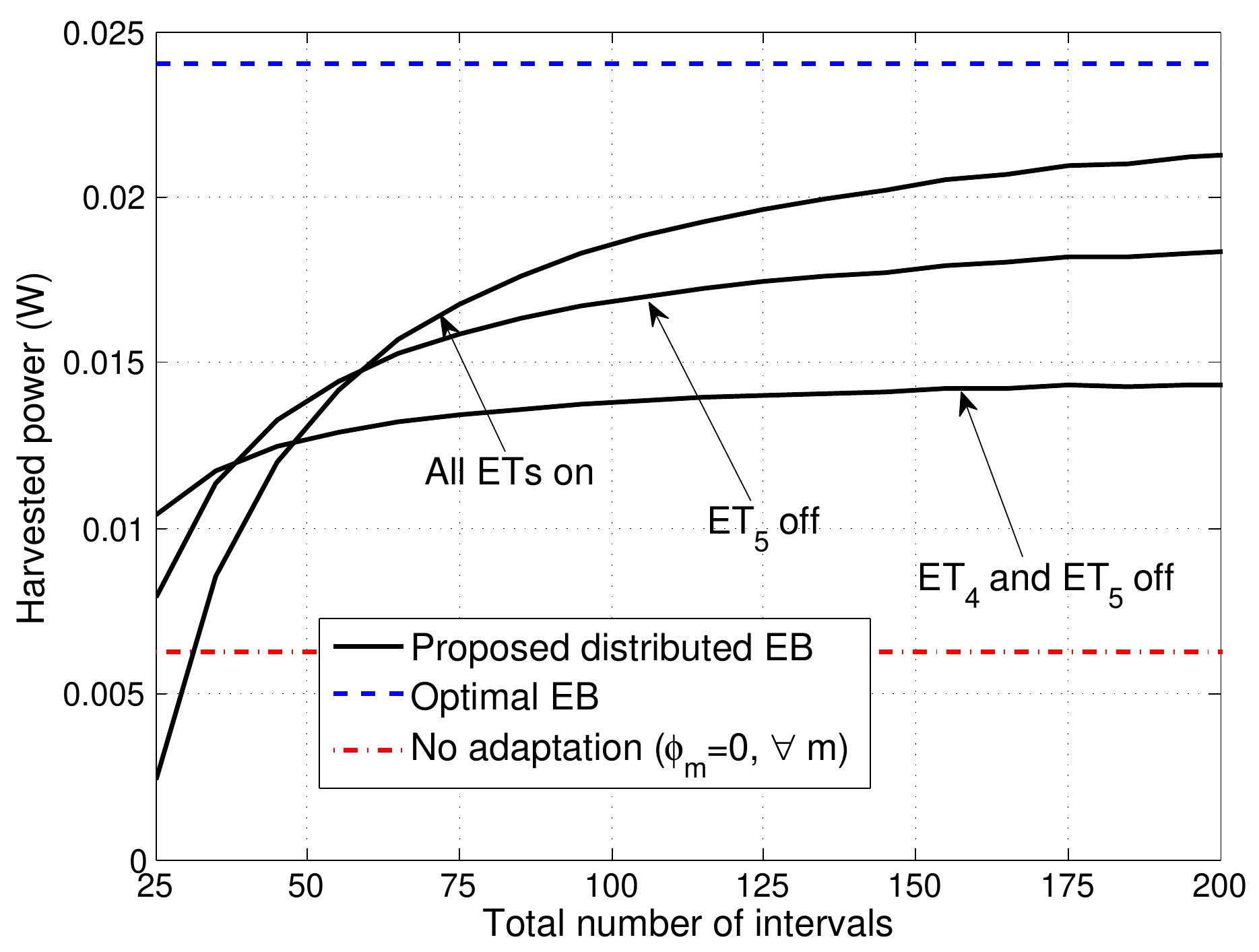}
\caption{The average harvested power versus the total number of intervals (including both training and energy transmission) with $M=5$ and $N=5$. } \label{Fig:TradeOff}
\end{figure}

Finally, in Fig.~\ref{Fig:TradeOff}, we plot the average harvested power per interval versus the total number of intervals (including both training and energy transmission shown in Fig.~\ref{Fig:Protocol}), by averaging over 50000 randomly generated $r_m$ and $\theta_m$, $m=1,...,M$, to analyze the effect of the training overhead on the performance of the proposed distributed EB scheme. Note that we assume $M=5$ and $N=5$ in this case. Moreover, by re-ordering the channel power gains of ETs as $\beta_1\geq\beta_2\geq...\geq\beta_5$, we also plot average harvested powers for the cases without ET$_5$ (i.e., the ET with the weakest channel to ER is off) or both ET$_5$ and ET$_4$ (i.e., the two ETs with lowest channel gains are off).  It can be seen from Fig.~\ref{Fig:TradeOff} that if the total number of intervals is less than 60, it is no more optimal to let all ETs transmit as the ETs with weak channels do not contribute much to the overall EB gain but require the same training time for phase adaptation (otherwise, their received signals may not add coherently to other ETs' received signals); thus, they should be switched off to maximize the average harvested power at ER.  Moreover, when the total number of intervals is less than 30, the performance with all ETs on is even worse than that of no adaptation with $\phi_m=0$, $m=1,...,M$. However, as the total number of intervals increases, the average harvested power also increases with more ETs switched on and finally approaches the maximum harvested power by the optimal EB with all ETs on, due to the reduced training overhead. 

\section{Conclusion}
In this paper, we proposed a new channel training and EB design for distributed WET systems, where ETs adjust their transmit phases independently to achieve collaborative WET to a single ER. Based on a new phase adaptation algorithm for each single ET to iteratively adapt its phase to match those of all other active ETs based on one-bit energy feedback from the ER, we devised  a distributed EB protocol, which is shown to converge to the optimal EB performance efficiently even with small number of per-ET feedback intervals. The proposed distributed EB scheme is also shown to outperform the existing scheme based on random phase perturbation in terms of convergence speed and energy efficiency.

\appendices 

\appendix[Proof of Proposition~\ref{Proposition:LowerBound}]

We need to show that the following inequality holds.
\begin{align} \label{Eq:Induction}
Q_\text{d} &\geq \sum_{m=1}^M \beta_m + \sum_{i=1,j=1, i\neq j}^M \sqrt{\beta_i \beta_j} \cos(e_i)\cos(e_j),
\end{align}
where $e_m$ is defined in \eqref{Eq:Error}. We prove \eqref{Eq:Induction} via mathematical induction as follows. First, we show that \eqref{Eq:Induction} holds for $M=2$. For convenience, we denote $Q_\text{d}$ given in \eqref{Eq:P_r Protocol} for $M=k$ as $Q_\text{d}^{(k)}$. In the case of $M=2$, when ET$_2$ adapts its phase $\phi_2$ via (A1), only ET$_1$ is active with $\bar{\phi}_1 = 0$ based on the protocol described in Section~\ref{Subsection:Protocol}, and thus the estimated phase at ET$_2$ is given by $\bar{\phi}_2 = \theta_2 - \theta_1 + e_2$. Thus, the harvested power at ER is given by
\begin{equation}
Q_\text{d}^{(2)} = \beta_1 + \beta_2 + 2\sqrt{\beta_1\beta_2}\cos(e_2),
\end{equation}
which implies that  \eqref{Eq:Induction} holds for $M=2$ with equality. Next, we assume that \eqref{Eq:Induction} holds for $M=k$, $k\geq 2$, i.e., the following inequality is true:
\begin{equation} \label{Eq:InductionAssumption}
Q_\text{d}^{(k)} \geq \sum_{m=1}^k \beta_m + \sum_{i=1,j=1, i\neq j}^k \sqrt{\beta_i \beta_j} \cos(e_i)\cos(e_j).
\end{equation}
Then, when $M=k+1$, the harvested power at ER by the distributed EB protocol is given by
\begin{align*}
Q_\text{d}^{(k+1)} &=  \beta_{k+1} + Q_\text{d}^{(k)} + 2 \cos(e_{k+1})\sqrt{\beta_{k+1}Q_\text{d}^{(k)}} \\
& \overset{(a)}\geq \beta_{k+1 } + \sum_{m=1}^k \beta_m + \sum_{i=1,j=1, k\neq l}^k \sqrt{\beta_i \beta_j} \cos(e_i)\cos(e_j)  \\
& \qquad + 2\cos(e_{k+1})\sqrt{\beta_{k+1}}\bigg(\sum_{m=1}^k \beta_m + \sum_{i=1,j=1, i\neq j}^k \sqrt{\beta_i \beta_j} \cos(e_i)\cos(e_j)\bigg)^{\frac{1}{2}}  \\
& \overset{(b)}\geq \sum_{m=1}^{k+1} \beta_m +  \sum_{i=1,j=1, i\neq j}^k \sqrt{\beta_i \beta_j} \cos(e_i)\cos(e_j) \\ 
&\qquad + 2\cos(e_{k+1}) \sqrt{\beta_{k+1}}\bigg(\sum_{m=1}^k \beta_m \cos^2(e_m)  +  \sum_{i=1,j=1, k\neq l}^k \sqrt{\beta_i \beta_j} \cos(e_i)\cos(e_j) \bigg)^{\frac{1}{2}} \\
& = \sum_{m=1}^{k+1} \beta_m +  \sum_{i=1,j=1, i\neq j}^{k+1} \sqrt{\beta_i \beta_j} \cos(e_i)\cos(e_j). 
\end{align*}
where $(a)$ is due to the assumption in \eqref{Eq:InductionAssumption}, and $(b)$ is due to the fact that $0\leq \cos^2 (e_m) \leq 1$, $m=1,...,k$.  To summarize, we have shown that \eqref{Eq:Induction} holds for $M=k+1$ under the assumption that it holds for $M=k$. Since we have already shown that \eqref{Eq:Induction} is true for $M=2$, we conclude that \eqref{Eq:Induction} holds for any $M\geq 2$.  

Next, to prove \eqref{Eq:EfficiencyLower}, we substitute the error bound given in \eqref{Eq:ErrorBound} into \eqref{Eq:Induction}. It thus follows that 
\begin{align} 
Q_d  \geq   \sum_{m=1}^M \beta_m + \sum_{i=1,j=1, i\neq j}^M \sqrt{\beta_i \beta_j} \cos^2\l(\frac{\pi}{2^N}\r), \label{Eq:EfficiencyLowerProof}
\end{align}
since $\cos (x)$ is a decreasing function in $0\leq x \leq \pi$. Finally, the desired result in \eqref{Eq:EfficiencyLower} can be obtained by substituting \eqref{Eq:EfficiencyLowerProof} into \eqref{Eq:Efficiency}. The proof is thus completed.

\bibliographystyle{IEEEtran}

\newpage

\end{document}